\documentclass{article}

\usepackage[utf8]{inputenc}
\usepackage{enumitem}
\usepackage{amssymb}
\usepackage{amsmath}
\usepackage{amsthm}
\usepackage{subfigure}

\usepackage{colortbl}
\usepackage{xcolor}
\usepackage{comment}
\usepackage{soul}
\usepackage{bm}

\usepackage{hyperref}

\usepackage{authblk}

\usepackage{graphicx}

\usepackage[ruled]{algorithm2e}

\newtheorem{theorem}{Theorem}
\newtheorem{lemma}{Lemma}

\newtheorem{definition}{Definition}

\newtheorem{problem}{Problem}

\DeclareMathOperator*{\argmax}{argmax}

\newcommand{\set}[1]{\left\{ #1 \right\}}
\newcommand{\card}[1]{\left| #1 \right|}

\newcommand{\Mbot}{\textbf{Multi\_Bot}}
\newcommand{\cmax}{c_{\mbox{\scriptsize{max}}}}
\newcommand{\dmax}{d_{\mbox{\scriptsize{max}}}}
\newcommand{\HMaxopt}{H^*}
\newcommand{\HMaxoptalign}{H^*}
\newcommand{\Xopt}{\mathbf{x}^*}
\newcommand{\Piopt}{\bm{\pi}^*}
\newcommand{\MbotOne}{\textbf{Multi\_Bot\_1P}}
\newcommand{\MbotOnel}[2]{\textbf{Multi\_Bot\_1P}[#1, #2]}
\newcommand{\IMS}{\textbf{IMS}}
\newcommand{\IMSlong}{\textbf{Identical-machines scheduling}}
\newcommand{\vol}{\mbox{vol}}
\newcommand{\volit}{\emph{vol}}
\newcommand{\sca}[1]{\mbox{scale}_{#1}}
\newcommand{\scait}[1]{\emph{scale}_{#1}}
\newcommand{\lpt}{\textsc{Lpt-first}}
\newcommand{\Ipi}{\mathcal{I}({\bm{\pi}})}
\newcommand{\tablebot}[1]{\texttt{Table}[#1]}
\newcommand{\boolbot}[1]{\texttt{Bool}[#1]}
\newcommand{\packbot}[1]{\texttt{Pack}[#1]}

\newcommand{\K}{\mathcal{K}}
\newcommand{\T}{\mathcal{T}}
\newcommand{\p}{\mathcal{P}}
\newcommand{\B}{\mathcal{B}}
\newcommand{\I}{\mathcal{I}}

\newcommand{\True}{\mbox{True}}
\newcommand{\False}{\mbox{False}}

\title{Approximation Algorithms for Job Scheduling with Reconfigurable Resources}
\author[1]{Pierre Berg\'e}
\author[2]{Mari Chaikovskaia}
\author[1]{Jean-Philippe Gayon}
\author[1]{Alain Quilliot}
\affil[1]{Université Clermont-Auvergne, CNRS, Mines de Saint-Etienne,
Clermont-Auvergne-INP, LIMOS, 63000 Clermont-Ferrand, France}
\affil[2]{IMT Atlantique, LS2N, UMR CNRS 6004, F-44307 Nantes, France}
\date{}

\begin{document}


\maketitle

\begin{abstract}
We consider here the \Mbot\ problem for the scheduling and the resource parametrization of jobs related to the production or the transportation of different products inside a given time horizon. Those jobs must meet known in advance demands. The time horizon is divided into several discrete identical periods representing each the time needed to proceed a job. The objective is to find a parametrization and a schedule for the jobs in such a way they require as less resources as possible. Though this problem derived from the applicative context of reconfigurable robots, we focus here on fundamental issues. We  show that the resulting strongly NP-hard \Mbot\  problem may be handled in a greedy way with an approximation ratio of $4/3$.  
\end{abstract}


\section{Introduction}

In many industrial contexts, automatized production must adapt itself to a fast evolving demand of a large variety of customized products. One achieves such a flexibility requirement through the notion of \textit{reconfiguration}. Once an operation is performed, related either to the production of some good or to its transportation, one may redesign the infrastructure that supported this operation by adding, removing or replacing some atomic components, or by modifying the links that connect those components together. Those components may be hardware (robots, instruments), software or human resources. They behave as renewable resources  \cite{becsikci2015multi,boysen2022assembly,hartmann2022updated} and move inside the production area in order to fit, during a given production cycle, with current production/transportation needs. Depending on the way one assigns those resources to a given operation, one may not only achieve this operation but also speed it, increase its throughput or lessen its cost, as in the \textbf{multi-modal Resource Constrained Project Scheduling Problem} \cite{becsikci2015multi}. 

We consider the following scheduling problem with reconfigurable resources. Several types of jobs (e.g. production operations, transportation tasks) have to be processed by a set of identical resources (e.g. workers, robots, processors) over a discrete time horizon in order to achieve a certain demand. Assigning a number $p$ of resources to some job of type $k$ gives a certain production $c_{pk}$: all production values, called \textit{capacities}, are given as inputs. In each time period, teams of resources must be formed to process jobs. 
A resource which is used to perform some type of job $k$ at period $t$ may be employed for another type of job $k'\neq k$ in the next period $t+1$.
The objective is to determine the minimum number of resources needed to obtain a certain production for each type of job. This problem is called \Mbot\ and is strongly NP-hard~\cite{chaikovskaia23thesis}. 

It has been firstly introduced in a warehouse logistics context  \cite{chaikovskaia2022sizing} in collaboration with the MecaBoTix company~\cite{mecabotix} which designs reconfigurable mobile robots. In this application, resources are mobile robots and jobs consist in moving loads of various types, such as pallets or boxes. Those robots are capable of aggregating into a cluster to form poly-robots that can adapt to the type of product (size/mass) and navigate independently in environments such as warehouses, production sites or construction sites. Other applications can be found in the automotive industry where the resources are workers and the processing time for a task in the assembly line depends on the number of workers assigned to it \cite{battaia2015workforce}. 

From a theoretical point of view, \Mbot\ is strongly related to some classical scheduling problems of the literature. A very well-known one is \IMSlong \, (\IMS), where the objective is to pack items into a set of identical boxes while minimizing the size of the most filled box~\cite{graham69}. The standard scheduling notation of \IMS\ is $P\vert \vert C_{\max}$. Its high multiplicity variant, denoted by $P\vert HM(n)\vert C_{\max}$, encodes as binary inputs the number of items with the same size~\cite{brauner05}. These two scheduling problems have been widely studied in terms of approximation and parameterized complexity~\cite{brinkop22,coffmangj78,hochbaum85,knop20,mccormick2001polynomial,mnich18,mnich15}. Problem $P\vert HM(n)\vert C_{\max}$, when the item sizes are polynomially bounded, is a special case of \Mbot : consider that any type of job $k$ correspond to some item size and can be performed only by a specific number of resource $p_k$, the demands of jobs of type $k$ correspond to the high-multiplicity coefficients of the related items. In this article, we will use a heuristic of \IMS\ as a sub-routine of our algorithm.

The main result of this paper is the presentation of a polynomial-time $\frac{4}{3}$-approximation algorithm for \Mbot . The impact of this contribution is, in our opinion, twofold. On one hand, we provide, for industrial applications such as MecaBotiX robots, a fast and efficient heuristic with the strict guarantee that it will not fail on pathological instances more than 33\% over the optimum. On the other hand, we extend the theoretical knowledge on approximability of scheduling problems, showing that a generalization of $P\vert HM(n)\vert C_{\max}$ admits a constant approximation ratio.

The paper is organized as follows. In Section \ref{sec:notation}, we detail the \Mbot\ problem, remind its ILP formulation and introduce the notions of \textit{schedule} and \textit{packing}. Next we provide in Section 3 the definition of \IMSlong\ (\IMS) problem as well as some preliminary observations on the optimum solutions of \Mbot . Section 4 is devoted to the description of our approximation algorithm. This algorithm relies on a polynomial time dynamic programming algorithm that solves  \Mbot\ when the periods are merged into a single macro-period: this is described in Section 5.

\section{Problem description and notations}\label{sec:notation}

We consider a set of identical resources (for instance robots, workers, etc) which cooperate in
order to process different types of jobs.  A $p$-resource is a \textit{configuration} which
makes $p$  resources cooperate on the same job. For example, in robotics, it models the fact that $p$ elementary robots can assemble together in order to transport heavy loads. A
maximum of $P$  resources may cooperate. The set of configurations is thus $\p$ = \{1,\ldots, $P$\}. There are $K$ job types and let $\K = \{1,\ldots, K\}$. The demand for type $k$ of jobs  is denoted by $d_k$ for every type $k \in \K$, and $\dmax$ is the maximum demand.  Given a job of type $k$, there is at least one value $p$ such that
a $p$-resource is able to process a job of type $k$.

All tasks must be executed within a discrete time horizon $\T = \{1,\ldots, T\}$. At the beginning of each period $t \in \T$, the resources may be
reconfigured in order to provide us with numerous configurations which perform jobs for
period $t$. During a given period $t$, a $p$-resource can deal with only a single type $k \in \K$, and its production is given by the capacity $c_{pk}$.
Our purpose is to  minimize the number of  resources involved
into the whole process. If $H_t$ denotes the number of active  resources
during period $t$, then the number of  resources necessary to achieve the whole process is $H = \max\limits_{t \in \T} H_t$. Let \textbf{Multi\_Bot} refer to this optimization problem.

\subsection{ILP formulation for \Mbot} \label{subsec:multibot}

More formally, we provide in Problem~\ref{pb:multibot} the integer linear programming (ILP) formulation of \Mbot.
We denote by $x_{pkt}$ the decision variable representing the number of jobs of type $k$ performed in configuration $p$ at period $t$. 

Constraint \eqref{ILP1} means that we must have a sufficient total capacity to satisfy demand $d_k$. Quantity  $\sum_{t\in\T} \sum_{p\in\p} c_{pk}\cdot x_{pkt}$ represents the maximum number of jobs of type $k$ that can be processed over the horizon, given the $x_{pkt}$. Constraint \eqref{ILP2} means that the number of required  resources in period $t$ can be written as $H_t$ = $\sum_{p,k} p \cdot x_{pkt}$. Constraint \eqref{ILP3} means that $H$, the number of  resources necessary to achieve the whole process, must be greater than or equal to every $H_t$. Constraint \eqref{ILP4} recalls that all decision variables are non negative integers.

Figure~\ref{fig:schedule} illustrates an optimal solution ({\em i.e.} which minimizes $H$) for the following instance of \Mbot . There are two types of jobs ($K = 2$) and three periods ($T=3$). The maximum size of a performed job is $P=5$. The demands are $d_1 = 13$ and $d_2 = 10$. The capacities for jobs of type 1 are $c_{11} = 1$ and $c_{21} = 4$. The demand $d_1 = 13$ is achieved as the schedule contains three $2$-resources (total production 12) and one $1$-resource (production 1) handling jobs of type $k=1$. The capacities for jobs of type 2 are linear in $p$: for any $1\le p\le P$, $c_{p2} = p$. One can check the demand for $k=2$ is also reached. Furthermore, it is not possible to find a solution which achieves $H = 5$.

\begin{problem}[\Mbot]
\begin{align}
 \mbox{\emph{\textbf{Input:}}} ~~ & \K = \set{1,\ldots,K}, \p = \set{1,\ldots,P} \nonumber \\
 &  \T = \set{1,\ldots,T} \nonumber \\
 & \emph{Capacities}~ (c_{pk})_{\substack{p \in \p \\ k \in \K}} \nonumber \\
 & \emph{Demands}~  (d_k)_{k \in \K} \nonumber\\
 \mbox{\emph{\textbf{Objective:}}} ~~   & \emph{minimize}  \, \,   H\label{ILP0}
 \\
        \mbox{\emph{\textbf{subject to:}}} ~~ & \sum_{t\in \T} \sum_{p\in \p} c_{pk} \cdot x_{pkt} \geq d_k & \forall k \in \K \label{ILP1}\\
     & H_{t} = \sum_{k\in\K} \sum_{p\in \p} p \cdot x_{pkt} & \forall t \in \T  \label{ILP2}\\
     & H \geq H_{t} & \forall t \in \T  \label{ILP3}\\ 
     & x_{pkt} , H \in \mathbb{N}   & \forall k \in \K, p \in \p, t\in \T    \label{ILP4}
\end{align}
\label{pb:multibot}
\end{problem}

\subsection{Schedules and packings} \label{subsec:schedules}

A solution of the \Mbot\ problem is thus a vector of size $PKT$: 
$$\mathbf{x}= (x_{pkt})_{\substack{p \in \p\\ k \in \K\\ t \in \T}}$$
In the remainder, we abuse notation when the context is clear: the same vector could be denoted implicitly by $(x_{pkt})_{p,k,t}$ to gain some space. Also, notations $x_{pkt}$ and $x_{p,k,t}$ refer to the same value.

\begin{figure}[t]
\centering
\includegraphics[scale=0.4]{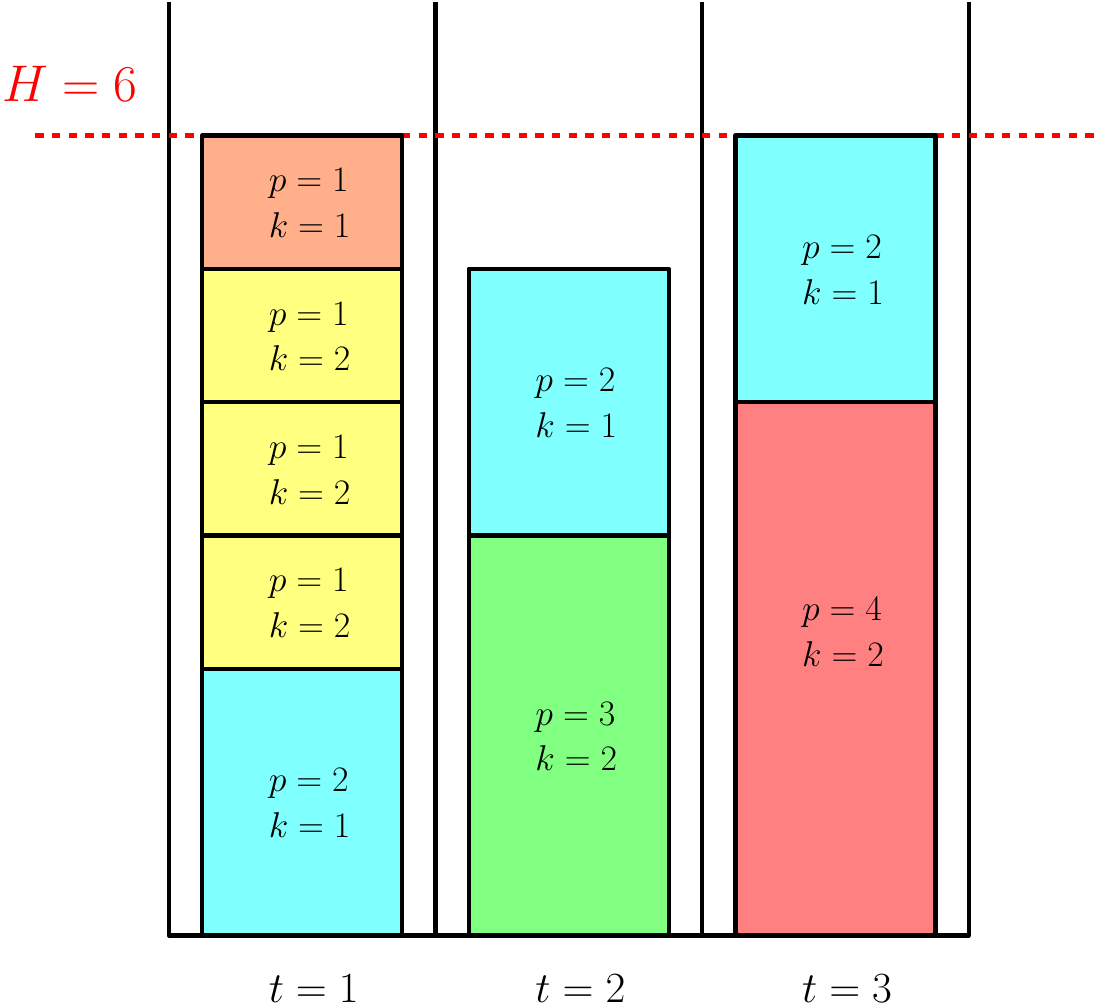}
\caption{An example of schedule $\mathbf{x} = (x_{pkt})_{p,k,t}$ with $H = \max_t H_t = 6$, meaning that at most $6$ resources are used per period.}
\label{fig:schedule}
\end{figure}

The input size of \Mbot\ is $O\left(T + KP(\log \cmax + \log \dmax)\right)$. Concretely, values $K,P,T$ can be seen as combinatorial inputs while demands and capacities are numerical values. Hence, values $c_{pk}$ and $d_k$ might be exponential in the input size, but also $x_{pkt}$, $H_t$ and $H$. However, observe that the size of a solution $\mathbf{x}$ is polynomial in the input size. We call such a solution a \textit{schedule} of the \Mbot\ instance.

Given an instance of \Mbot,
\begin{itemize}
\item let $\HMaxopt$ denote the optimum value for the evaluation function $H$,
\item let $\Xopt = \left(x_{pkt}^*\right)_{p,k,t}$ be a schedule reaching this optimum.
\end{itemize} 

In comparison with a schedule $\mathbf{x}$ which is a vector with $KPT$ values, a \textit{packing} $\bm{\pi} = (x_{pk})_{\substack{p \in \p\\ k \in \K}}$  represents the schedule of jobs of type $k$ by $p$-resources, independently from any time consideration. For example, it can be used to describe the production process during one specific period. This is a vector with $KP$ values. Given some schedule $\mathbf{x}= (x_{pkt})_{p,k,t}$, its \textit{associated packing} is simply given by all values $\sum_{t\in\T} x_{pkt}$, for all $p\in \p, k\in \K$. In particular, the packing associated with the optimal schedule $\Xopt$ is denoted by $\Piopt$ :
$$\Piopt = \left(\sum_{t\in \T} x_{pkt}^*\right)_{\substack{p \in \p\\ k \in \K}}$$

We pursue with the definition of measures for both schedules and packings.

\begin{definition}[Volume]
The \emph{volume} of a packing is the total number of resources involved in this packing. Formally, for $\bm{\pi} = (x_{pk})_{p,k}$,
$$\volit(\bm{\pi}) = \sum_{k\in \K} \sum_{p\in \p} p\cdot x_{pk}$$
For the sake of simplicity, we also use this notion for schedules. The volume of a schedule $\mathbf{x}$ is the volume of its associated packing, {\em i.e.} $\sum_{p,k,t} p\cdot x_{pkt}$.
\label{def:volume}
\end{definition}
\begin{definition}[Maximum]
The \textit{maximum} of a packing $\bm{\pi}$ is the maximum $p \in \mathcal{P}$ such that some $x_{pk}$ is non-zero. Formally, $$\max(\bm{\pi}) = \max \set{p \in \mathcal{P} : x_{pk} > 0 ~\mbox{for some}~ k}$$.
\label{def:max}
\end{definition}

Eventually, we define another measure on packings that we will use in our approximation algorithm: the \textit{scale}.

\begin{definition}[Scale]
Given some integer parameter $\lambda$, let us call the \emph{big configurations} the $p$-resources with $p>\frac{2\lambda}{3}$ and the \emph{medium configurations} the $p$-resources with $\frac{\lambda}{3}<p\le \frac{2\lambda}{3}$. Naturally, the \emph{small configurations} refer to $p \le \frac{\lambda}{3}$. We define the $\lambda$-\emph{scale} as the number of jobs performed with big configurations in packing $\bm{\pi}$ plus half the number of jobs performed with medium configurations in $\bm{\pi}$. It is a half-integer:
\begin{equation}
\sca{\lambda}(\bm{\pi}) = \sum_{k\in \K} \sum_{p > \frac{2\lambda}{3}} x_{pk} + \frac{1}{2}\sum_{k\in \K} \sum_{\frac{\lambda}{3} < p \le \frac{2\lambda}{3}} x_{pk}
\label{eq:scale}
\end{equation}
\end{definition}

Given a solution $\mathbf{x}$ of \Mbot\ using $H$ resources, its associated packing $\bm{\pi}$ has a limited $H$-scale. Indeed, looking at how much medium or big configurations a period can contain, we see that it has at most either one big configuration or two medium configurations. For example, a period cannot contain both a big and a medium configuration, by definition. Hence, $\sca{\lambda}(\bm{\pi}) \le T$.

As an example, for the packing $\bm{\pi}$ associated to the schedule proposed in Figure~\ref{fig:schedule}, we have $\vol(\bm{\pi}) = 17$, $\max(\bm{\pi}) = 4$ and $\sca{5}(\bm{\pi}) = 1 + 4\cdot \frac{1}{2} = 3$.

\section{Preliminaries }

\subsection{Approximation algorithms}

Unfortunately, \Mbot\ is strongly NP-complete for the general case because it can be reduced from \textbf{Bin Packing}~\cite{chaikovskaia23thesis}. Consequently, assuming P$\neq$NP, there is no polynomial-time exact algorithm for \Mbot, even if the numerical values are supposed to be polynomially-bounded by the input size. A very natural question is thus the approximability of this problem.

An $r$-approximation algorithm, $r\ge 1$, for \Mbot\ is a polynomial-time algorithm which outputs a solution $\mathbf{x} = (x_{pkt})_{p,k,t}$ such that:
$$H = \max\limits_{t \in \T} \left(\sum\limits_{k \in \K} \sum\limits_{p \in \p} p \cdot x_{pkt}\right) \le r \cdot \HMaxopt$$

Approximation algorithms offer the guarantee that the number of resources used by the solution returned is at most a linear function of the optimum.

\subsection{Identical-machines scheduling}

We recall the definition and some results related to a well-known problem in operations research: \IMSlong\ \cite{graham69}. This problem can be seen as the optimization of \textbf{Bin Packing} by the capacities, where the number of boxes is fixed. In the scheduling framework, \IMS\ corresponds to $P\vert \vert C_{\max}$. Its objective is to assign a set of $n$ tasks given with their processing times to $m$ identical machines such that the makespan is minimized. To distinguish \IMS\ with \Mbot , we will use a slightly different syntax. Formally, we are given a set of items $\I=\{1,\cdots, n\}$ and a set of boxes $\B = \set{1,\ldots,m}$. An item $i$ has a certain size $s_i$. The objective is to pack all items into the boxes such that the total size packed inside the different boxes is balanced. More precisely, we aim at minimizing the size of the most filled box.

\begin{problem}[\textbf{Identical-machines scheduling} (\IMS)]
\begin{align}
 \mbox{\emph{\textbf{Input:}}} ~~ & \emph{Items}~ \I = \set{1,\ldots,n},  \nonumber  \\
 & \emph{Boxes}~ \B = \set{1,\ldots,m} \nonumber \\
 & \emph{Sizes}~ \set{s_{i}}_{i\in \I} \nonumber \\
 \mbox{\emph{\textbf{Objective:}}} ~~   & \emph{minimize}  \, \,   H
 \\
        \mbox{\emph{\textbf{subject to:}}} ~~ & \sum_{j\in \B} x_{ij} = 1 & \forall i\in \I \label{IMS1}\\
     & H_{j} = \sum_{i\in \B} p_i\cdot x_{ij}  & j \in \B \label{IMS2}\\
     & H \geq H_{j} & j \in \B \label{IMS3}\\ 
     & x_{ij} \in \set{0,1}& i \in \I, j \in \B  \label{IMS4}
\end{align}
\label{pb:ims}
\end{problem}

Observe that the minimization function of \Mbot\ and \IMS\ are similar: in both problems, we aim at minimizing a certain ``volume'' of the jobs/items which have been put into the most filled box/period. Naturally, we will try in the remainder to reduce - in some sense - instances of \Mbot\ into instances of the well-known problem \IMS. There is a natural correspondence between packings and instances of \IMS , since each $p$-resource perfoming a job $k$ in $\bm{\pi}$ can be seen as an item of size $p$. Said differently, for any $p \in \mathcal{P}$, the $\sum_k x_{pk}$ resources present in packing $\bm{\pi}$ can be converted into $\sum_k x_{pk}$ items of size $p$.

Given an \IMS\ instance $\mathcal{J}$, we define:
\begin{itemize}
    \item its \textit{volume} $\vol(\mathcal{J})$ as the total size of its items, {\em i.e.} $\vol(\mathcal{J}) = \sum_{i\in \I} s_i$,
    \item its \textit{maximum} as the maximum item size: $\max(\mathcal{J}) = \max_{i\in \I} s_i$.
\end{itemize}

A heuristic of \IMS\ will be used as a sub-routine for our approximation algorithm dedicated to \Mbot . It is called \textsc{Longest-processing-time-first} (\lpt) ~\cite{graham69}. Its description is relatively simple: first sort the items by decreasing order of their sizes (the largest size comes first), second put the items into the boxes following this order. The filling must satisfy the following rule: we always fill the box with the largest empty space, or said differently the least filled box. Figure~\ref{fig:lpt} shows the packing obtained with \lpt\ for a certain instance. Graham showed that \lpt\ is a $\frac{4}{3}$-approximation algorithm~\cite{graham69}. In addition to offering a small approximation ratio, \lpt\ is a fast algorithm: it runs only in $O(n(\log m + \log n))$.

\begin{figure}
\centering
\includegraphics[scale=0.5]{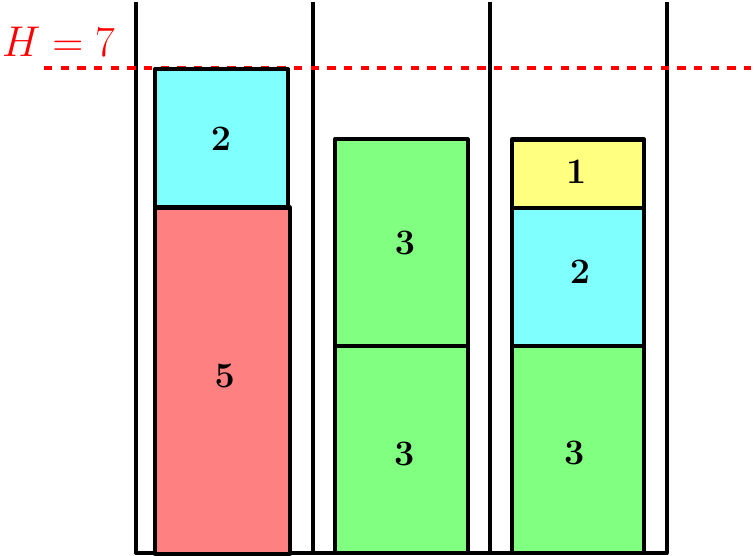}
\caption{Output of \lpt\ with item sizes $(5,3,3,3,2,2,1)$}
\label{fig:lpt}
\end{figure}

In the literature, \IMS\ admits other approximation algorithms, such as \textsc{Multifit}~\cite{coffmangj78} which has a ratio $\frac{13}{11}$, proven by Yue~\cite{yue90}. Hochbaum and Schmoys proposed a PTAS~\cite{hochbaum85} and we know that \IMS\ does not admit a FPTAS since it is strongly NP-complete~\cite{garey1979computers}. In the remainder, despite the existence of smaller approximation ratios, we focus only on the \lpt\ algorithm since it allows us to identify an approximation of \Mbot\ with our framework.

\subsection{Optimal configurations}

We provide a crucial observation for \Mbot\ which will help us designing efficient algorithms in the remainder. This is based on the notion of \textit{optimal configuration} described below.

\begin{definition}[Optimal configuration for $k \in \K$]
For any $k \in \K$, let $p_0(k)$ be the \emph{optimal configuration} for jobs of type $k$, {\em i.e.} the configuration $p \in \p$ which is the most efficient in terms of resources. In brief,
$$p_0(k) = \argmax_{p \in \p} \frac{c_{pk}}{p}$$
\label{def:optim_config}
\end{definition}

For any optimum solution $\Xopt$ of \Mbot , given some type $k$ of job, the number $x_{p_0(k),k,t}^*$ of $p_0(k)$-resources used at each period $t$ is potentially very large (exponential in the input size). However, we know that the number of $p$-resources used, with $p \neq p_0(k)$, is bounded polynomially by the input size.

\begin{lemma}
There is an optimum solution $\Xopt$ for which, for any $k \in \K$, and $p \neq p_0(k)$, $x_{pkt}^* \le p_0(k) \le P$.
\label{le:optim_config}
\end{lemma}
\begin{proof}
Consider an arbitrary optimum schedule $\Xopt$ and assume that some $x_{pkt}^*$, with $k \in \mathcal{K}$ and $p \neq p_0(k)$, is larger than $p_0(k)$. Then, we build a schedule $\mathbf{x}'$ with $H' = \max_t \sum_{p,k} p x_{pkt}'$ at most $\HMaxopt$ and which satisfy all demands.  We fix $q = \lfloor \frac{x_{pkt}^*}{p_0(k)} \rfloor \ge 1$. Let $\mathbf{x}'$ be the same schedule than $\Xopt$ but replace, at period $t$, a number $p_0(k)q$ of $p$-resources processing jobs of type $k$ by a number $pq$ of $p_0(k)$-resources processing also jobs of type $k$. Observe that this transformation does not modify the volume taken for period $t$. Moreover, the production is only modified for jobs of type $k$ and, as we used a more efficient configuration in the new schedule, the production of jobs of type $k$ at period $t$ cannot decrease. Formally, $$\sum_{p'=1}^P c_{p'k}x'_{p'k} = \sum_{p'=1}^P c_{p'k}x^*_{p'k} + q\left(pc_{p_0(k),k} - p_0(k)c_{pk}\right) = \sum_{p'=1}^P c_{p'k}x^*_{p'k} + R$$
By Definition~\ref{def:optim_config}, we know that $c_{p_0(k),k} \ge p_0(k)\frac{c_{pk}}{p}$, so $R\ge 0$. In summary, schedule $\mathbf{x}'$ ensures at least the same production than $\Xopt$, uses the same volume per period (same $H^*$). Therefore, it is also an optimum schedule and it satisfies $x_{pkt}' \le p_0(k)$.
\end{proof}

From now on, any optimum solution $\Xopt$ of \Mbot\ will be supposed to be such as the one described by Lemma~\ref{le:optim_config}. 

\section{Structure and analysis of the approximation algorithm}

In this section, we present the shape of our approximation algorithm. This will consist in two steps: first solving a slightly different problem from \Mbot\ with only one period, second use its solution to propose a global schedule for the $T$ periods. We show how the approximation ratio $\frac{4}{3}$ can be obtained for the general \Mbot\ when one-period problems are solved exactly.

\subsection{Presentation of the algorithm}

Our idea to design approximations algorithms for the general \Mbot\ follows. We call this general framework \textsc{Bot-approx}:
\begin{itemize}
    \item \textbf{Step 1.} Compute a polynomial-sized collection $\Pi$  of packings with structural properties (see Theorem~\ref{th:collection} for details),
    \item \textbf{Step 2.} For each $\bm{\pi} \in \Pi$, create an \IMS\ instance $\Ipi$ with $T$ boxes and items which are directly obtained from the packing $\bm{\pi}$ (transformation $\bm{\pi} \rightarrow I(\bm{\pi})$ is described in  Section \ref{subsec:groups}),
    \item \textbf{Step 3.} Find an approached solution for all \IMS\ instances $\Ipi$ with \textsc{lpt-first}. Return the one which corresponds to the best schedule.
\end{itemize}

Step 1 is very fuzzy for now, and we will introduce in Section~\ref{sec:packings} our method for computing this collection of packings. Our objective is to produce a collection $\Pi$ containing at least one packing $\bm{\pi} \in \Pi$ which is a good candidate for achieving a satisfying schedule $\mathbf{x}$ when we put its $p$-resources into the periods. Indeed, at least one packing $\bm{\pi} \in \Pi$ will allow us to return after Steps 2 and 3 a schedule $\mathbf{x}$ with $H \le \frac{4}{3}H^*$. The properties of this collection $\Pi$ are presented in Theorem~\ref{th:collection}. In particular, all packings of $\Pi$ will satisfy the demands $d_k$ for each type of job $k$, \emph{i.e.} $\sum_{p\in \p} c_{pk}x_{pk} \ge d_k$.  

\begin{theorem}\label{th_producing_packing}
Given some \Mbot\ instance, we can produce in polynomial time $O(KP^3T^3 + KP^7)$ a collection $\Pi$ of at most $3P$ packings such that:
\begin{itemize}
    \item each packing in $\Pi$ satisfy all demands,
    \item it contains at least one packing $\bm{\pi}$ which satisfies $\volit(\bm{\pi}) \le \volit(\Piopt)$,
    \item if $H^* < 3P$, it contains at least one packing $\bm{\pi}'$ which satisfies $\volit(\bm{\pi'}) \le \volit(\Piopt)$, $\max(\bm{\pi'}) \le H^*$, and $\scait{H^*}(\bm{\pi}') \le T$.
\end{itemize}
\label{th:collection}
\end{theorem}
\begin{proof}
Section~\ref{sec:packings} is entirely dedicated to the proof of this result. Its conclusion is given in Section~\ref{subsec:proof}.
\end{proof}

In the remainder of this algorithm, we apply Steps 2 and 3 for any packing in $\Pi$. Eventually, as the size of the collection is at most $3P$, we will obtain a set of at most $3P$ schedules. We will simply keep the one which provides the minimum $H$.

Step 2 will be detailed in Section~\ref{subsec:groups}. For some packing $\bm{\pi}\in \Pi$, a natural idea is, for each value $x_{pk}$, $p \in \mathcal{P}$, $k \in \mathcal{K}$, to create $x_{pk}$ items of size $p$ and solve \IMS\ with $T$ boxes. Said differently, we can construct an instance, with exactly $\sum_k x_{pk}$ items of size $p$ for each $p \in \mathcal{P}$, which will be equivalent to the packing $\bm{\pi}$. In this way, a solution of this \IMS\ instance with $T$ boxes correspond to a schedule for the initial \Mbot\ instance. 

Unfortunately, this transformation might not be achieved in polynomial time as values $x_{pk}$, which depend on demands $d_k$, can be exponential in the input size of \Mbot . In other words, we would create an \IMS\ instance of exponential size. Hence, we present a polynomial-time method which allows us to handle this issue and produce a polynomial-sized \IMS\ instance for $\bm{\pi}$, denoted by $\Ipi$. 

Eventually, Step 3 consists in applying \textsc{lpt-first} on each \IMS\ instance $\Ipi$ -  created at Step 2 - with $T$ boxes. The solutions obtained thus correspond to schedules, as the $T$ boxes of \IMS\ represent the periods of \Mbot\ and each of these periods contains a set of performed jobs (represented by the items), characterized by a configuration $p \in \p$ ans some type of job $k \in \K$. Consequently, the output of the whole process is a collection of $\card{\Pi} \le 3P$ schedules $\mathbf{x}$. Naturally, we keep the schedule with the minimum $H$. We will show that it offers a $\frac{4}{3}$-approximation for \Mbot\ (see Theorems~\ref{th:approx_1} and~\ref{th:approx_2}).

Based on the Theorems~\ref{th:collection}, \ref{th:approx_1} and~\ref{th:approx_2} cited above and proved in the remainder of the article, we present the main result of this paper.

\begin{theorem}[Approximation ratio of \textsc{Bot-approx}]
\textsc{Bot-approx} is a $\frac{4}{3}$-approximation algorithm for \emph{\Mbot}.
\label{th:main}
\end{theorem}
\begin{proof}
Consider some instance of \Mbot . We distinguish two cases, depending on the value of $\HMaxopt$.

If $\HMaxopt \ge 3P$, then we know that the collection $\Pi$ computed at Step 1 contains a packing $\bm{\pi}$ with a smaller volume than $\Piopt$ (Theorem~\ref{th:collection}). According to Theorem~\ref{th:approx_1}, the schedule $\mathbf{x}$ produced with \lpt\ by considering this initial packing $\bm{\pi}$ offers the guarantee that $H \le \frac{4}{3} \HMaxopt$.

If $\HMaxopt < 3P$, then, again from Theorem~\ref{th:collection}, collection $\Pi$ contains a packing $\bm{\pi}'$ with a smaller volume than $\Piopt$ such that $\max(\bm{\pi}') \le H^*$ and $\sca{H^*}(\bm{\pi}') \le T$. By Theorem~\ref{th:approx_2}, the schedule $\mathbf{x'}$ produced with \lpt\ by considering this initial packing $\bm{\pi}'$ gives also $H' \le \frac{4}{3} \HMaxopt$.

As \textsc{Bot-approx} returns the schedule minimizing $H$ among all initial packings $\bm{\pi} \in \Pi$, we are sure to obtain a final solution which uses at most $\frac{4}{3} \HMaxopt$ resources.
\end{proof}

From now on, in the next sections, our objective is to prove Theorem~\ref{th:collection} (Section~\ref{sec:packings}), but also Theorems~\ref{th:approx_1} and~\ref{th:approx_2} (Section~\ref{subsec:meta}) which are the keystones for the proof of Theorem~\ref{th:main}.

\subsection{Transformation of a packing into polynomial \IMS} \label{subsec:groups}

We assume now that we are working with some given packing $\bm{\pi}$, which belongs to the collection computed with Theorem~\ref{th:collection} and is a good candidate to obtain an approximate solution for the \Mbot\ instance. The objective is to ``schedule'' this packing into $T$ periods. Packing $\bm{\pi}$ satisfies all demands $d_k$ of the \Mbot\ instance. In order to assign efficiently each performed job of $\bm{\pi}$ into the periods, we model it as an \IMS\ instance and then use \lpt\ to ensure the approximation factor. We focus in this subsection on the transformation from packing $\bm{\pi}$ to an \IMS\ instance $I(\bm{\pi})$.

A natural way to achieve such equivalent transformation is simply, for each pair $k \in \mathcal{K}, p \in \mathcal{P}$, to create $x_{pk}$ items of size $p$. In this way, we obtain an \IMS\ instance with a volume ({\em i.e.} total size of the items) equal to $\vol(\bm{\pi})$. Each item of $I(\bm{\pi})$ thus represents a performed job and its size gives us the number $p$ of resources it involves. Furthermore, we keep in memory, for each item (equivalently for each $p$-resource of packing $\bm{\pi}$), which type of jobs this $p$-resource performs, even if it has no impact on the instance $I(\bm{\pi})$. Hence, assigning these items to a period $t \in \mathcal{T}$ produces a solution of \IMS\ which completely corresponds to a schedule, since it is equivalent to assigning $p$-resources performing jobs of type $k$ to period $t$, {\em i.e.} proposing some vector $\left(x_{pkt}\right)_{p,k,t}$.

Unfortunately, this simple transformation $\bm{\pi} \rightarrow I(\bm{\pi})$ can produce an exponential-sized instance.
As mentioned in the previous subsection, the values $x_{pk}$ might be exponential in the input size of \Mbot, as they depend on capacities and demands which are numerical values. To avoid  an \IMS\ instance of exponential size, our idea consists in forming ``large'' items which will represent a set of $p$-resources, instead of a single one. In fact, we will distinguish two cases. When $\vol(\bm{\pi})$ is upper-bounded by some polynomial function of $P$ and $T$ given below, we use the natural process of transforming each value $x_{pk}$ into $x_{pk}$ items of size $p$. Otherwise, each item will correspond to a set of $p$-resources and the polynomial size of the constructed \IMS\ instance will be guaranteed. The formal definition follows.

\begin{definition}\label{def_IMS}
We define $I$ as a polynomial-time algorithm which given some packing $\bm{\pi} = (x_{pk})_{p,k}$, produces an \IMS\ instance with the following rules:
\begin{itemize}
    \item If $\volit(\bm{\pi}) \le 3PT$, for each pair $p \in \mathcal{P}, k\in \mathcal{K}$, add $x_{pk}$ items of size $p$ into instance $I(\bm{\pi})$.
    \item Otherwise, if $\volit(\bm{\pi}) > 3PT$, items will represent a set of at most $\frac{\volit(\bm{\pi})}{3T}$ performed jobs with the same configuration $p$ and performing the same type $k$ of jobs. Analytically, for each $p$, let $\alpha_{p} = \lfloor \frac{\volit(\bm{\pi})}{3pT} \rfloor$. For each pair $p \in \mathcal{P}, k\in \mathcal{K}$, add $\lfloor \frac{x_{pk}}{\alpha_p} \rfloor$ items of size $p\alpha_p$, and 1 item of size $p(x_{pk} - \lfloor \frac{x_{pk}}{\alpha_p} \rfloor \alpha_{p})$
\end{itemize}
\label{def:transfo_poly}
\end{definition}

The first case, $\vol(\bm{\pi}) \le 3PT$, corresponds to the natural transformation described above, so we do not give more details on it. However, the second one, $\vol(\bm{\pi}) > 3PT$, needs extra explanations. Here, an item represents a ``block'' of several $p$-resources performing jobs of types $k$, in order to ensure that $I(\bm{\pi})$ has a polynomial size in the encoding of the initial \Mbot\ instance. Each item is thus associated with a configuration $p$ and a type of job $k$. Considering some pair $(p,k)$, almost all items associated with it (except one) represent a number $\alpha_p$ of $p$-resources, so their size is $p\alpha_p$. Their number is given by the division of $x_{pk}$ (number of $p$-resources performing jobs of type $k$ in packing $\bm{\pi}$) by the number $\alpha_p$ of $p$-resources represented by each block. But, if $x_{pk}$ is not a multiple of $\alpha_p$, an extra item should be added into $I(\bm{\pi})$ to represent the remaining $p$-resources.

Observe that, in both cases, the total volume of the \IMS\ instance $I(\bm{\pi})$ is equal to the volume of packing $\bm{\pi}$ as we represented all the performed jobs into $I(\bm{\pi})$. 
But the crucial property ensured by transformation $I$ is certainly that the returned \IMS\ instance has a size polynomial in $P,K,T$. 

\begin{lemma}
Let $\bm{\pi}$ be some packing. The \IMS\ instance $I(\bm{\pi})$: 
\begin{enumerate}
    \item contains $O(PKT)$ items,
    \item satisfies $\volit(I(\bm{\pi})) = \volit(\bm{\pi})$
    \item if $\volit(\bm{\pi}) > 3PT$ and $\volit(\bm{\pi}) \le \volit(\Piopt)$, does not contain items of size greater than $\frac{H^*}{3}$. 
\end{enumerate}
\label{le:poly_input}
\end{lemma}
\begin{proof}
1. If $\vol(\bm{\pi}) \le 3PT$, then the number of items is exactly $\sum_{p,k} x_{pk}$ which is smaller than $\sum_{p,k} p\cdot x_{pk} = \vol(\bm{\pi}) \le 3PT$. Else, we know that for each pair $(p,k)$, there are at most $\frac{x_{pk}}{\alpha_p} +1$ items. Let $\mu_p = \frac{\vol(\bm{\pi})}{3pT}$. As $\mu_p > 1$, $\alpha_p = \lfloor \mu_p \rfloor \ge \frac{\mu_p}{2}$. Hence, $$\frac{x_{pk}}{\alpha_p} = \frac{x_{pk}}{\lfloor \mu_p \rfloor}\le \frac{2x_{pk}}{\mu_p} = 6T\frac{p\cdot x_{pk}}{\vol(\bm{\pi})} \le 6T.$$
Consequently, the total number of items in $I(\bm{\pi})$ is at most $PK(6T+1)$.

2. The volume of an \IMS\ instance is the sum of the sizes over all its items. If $\vol(\bm{\pi}) \le 3PT$, then $\vol(I(\bm{\pi})) = \sum_{p,k} p\cdot x_{pk} = \vol(\bm{\pi})$. Otherwise, $$\vol(I(\bm{\pi})) = \sum_{p,k} \left(p\alpha_p\lfloor \frac{x_{pk}}{\alpha_p} \rfloor + p(x_{pk} - \lfloor \frac{x_{pk}}{\alpha_p} \rfloor \alpha_{p}) \right) = \sum_{p,k} p\cdot x_{pk} = \vol(\bm{\pi}).$$

3. If $\vol(\bm{\pi}) > 3PT$, then $p\alpha_p = p \lfloor \frac{\vol(\bm{\pi})}{3pT} \rfloor \le \frac{\vol(\bm{\pi})}{3T}$: as the volume of $\bm{\pi}$ is at most the one of $\Piopt$ and $H^*\ge \frac{\vol(\Piopt)}{T}$, it gives $p\alpha_p \le \frac{H^*}{3}$. 
\end{proof}

\subsection{Approximation analysis} \label{subsec:meta}

The objective of this subsection is to show that, given the \IMS\ instances $\Ipi$ we created with Steps 1 and 2, \lpt\ algorithm provides a $\frac{4}{3}$-approximation for \Mbot\ on at least one of these instances. Indeed, as each $\Ipi$ contains $T$ boxes, which can be seen as the periods of \Mbot , any of its solution can be directely converted into a schedule $(x_{pkt})_{p,k,t}$.

We begin with the proof of Theorem~\ref{th:approx_1}. Given some packing $\bm{\pi}$ with a smaller volume than $\Piopt$ and which satisfy all demands, \lpt\ achieves a $\frac{4}{3}$-approximation ratio under the condition: $H^* \ge 3P$.

\begin{theorem}
Consider some \emph{\Mbot} instance and a packing $\bm{\pi}$ which satisfies all demands. Moreover, $\volit(\bm{\pi}) \le \volit(\Piopt)$. If $\HMaxoptalign \ge 3P$, then the schedule $\mathbf{x}$ returned by solving $\Ipi$ with \lpt\ verifies $H \le \frac{4}{3} \HMaxoptalign$.
\label{th:approx_1}
\end{theorem}
\begin{proof}
Let $N$ be the number of items in $I(\bm{\pi})$. We know that $N = O(PKT)$ from Lemma~\ref{le:poly_input}. We proceed by induction on the number $z$ of items packed by \lpt . We prove that, for any $0\le z\le N$, the number of resources present in each period (or box) after packing the $z$ most large items is at most $\frac{4}{3}H^*$.

The base case is trivial: when $z = 0$, no item was packed, so the current $H$ is zero. Now, assume we already packed the $z$ first items and we want to pack item $z+1$. According to Lemma~\ref{le:poly_input}, $\vol(I(\bm{\pi})) = \vol(\bm{\pi}) \le \vol(\Piopt)$, therefore there is necessarily a period which contains less than $H^*$ resources, otherwise the total volume would overpass $TH^*$, a contradiction. So, the least filled period contains at most $H^*$ resources. Refering to the transformation $I$, see Definition~\ref{def:transfo_poly}, if $\vol(\bm{\pi}) \le 3PT$, each item represents a $p$-resource, so we pack item $z+1$ of size at most $P \le \frac{H^*}{3}$ into it. The volume of this period after adding item $z+1$ is at most $\frac{4}{3}H^*$. If $\vol(\bm{\pi}) > 3PT$, we also pack some item of size at most $p\alpha_p \le \frac{H^*}{3}$ according to Lemma~\ref{le:poly_input}. Together with the induction hypothesis, this observation shows us that all periods use at most $\frac{4}{3}H^*$ resources after packing $z+1$ items: $H\le \frac{4}{3}H^*$.
\end{proof}

Now we fix the other side of the tradeoff, {\em i.e.} when $H^* < 3P$.
\begin{theorem}
Consider some \Mbot\ instance with $\lambda = \HMaxoptalign$ and a packing $\bm{\pi}$ which satisfy all demands. Moreover, $\volit(\bm{\pi}) \le \volit(\Piopt)$, $\max(\bm{\pi}) \le \lambda$ and $\scait{\lambda}(\bm{\pi}) \le T$.
If $H^* < 3P$, the schedule $\mathbf{x}$ returned by solving $\Ipi$ with \lpt\ verifies $H \le \frac{4}{3} \HMaxoptalign$.
\label{th:approx_2}
\end{theorem}
\begin{proof}
As $H^* < 3P$, $\vol(\bm{\pi}) \le \vol(\Piopt) \le TH^* < 3PT$, hence the transformation $I$ simply consists, for each pair $(p,k)$, in adding $x_{pk}$ items of size $p$. As a consequence, each item represents exactly one $p$-resource of packing $\bm{\pi}$.

\lpt\ packs first the items of sizes $p > \frac{2H^*}{3}$. We look at the packing after adding only these big items.
As $\lambda = H^*$ and $\sca{\lambda}(\bm{\pi}) \le T$, then there are at most $T$ big items. Moreover, all big items ($p > \frac{2H^*}{3}$) do not overpass a size $H^*$, as $\max(\bm{\pi}) \le \lambda$. So, at this moment, at most 1 item is present in each period, their size is at most $H^*$ resources, and the periods which are not filled with one big item are empty.

Second, \lpt\ packs the medium items of sizes $p > \frac{H^*}{3}$ which are not big. The empty boxes are filled and, when none of them is empty anymore, the remaining medium items are packed upon another medium item, as their boxes are less filled than the ones with big items. There cannot be three medium items into one period $t$ and only one medium item into another period $t'$ since the volume of two medium items is at least $\frac{2H^*}{3}$ and is necessarily larger than for exactly one medium item. Consequently, if there is a period with at least three medium items, then all other periods containing medium items are made up of two of them at least. However, having such a period with three medium items contradicts the scale criterion as we would have $\sca{\lambda}(\bm{\pi}) > T$. Consequently, there cannot be more than two medium items per period, and the number of resources present in the periods filled with medium items is at most $\frac{4H^*}{3}$.

\begin{figure}[h]
    \centering
    \includegraphics[scale = 0.7]{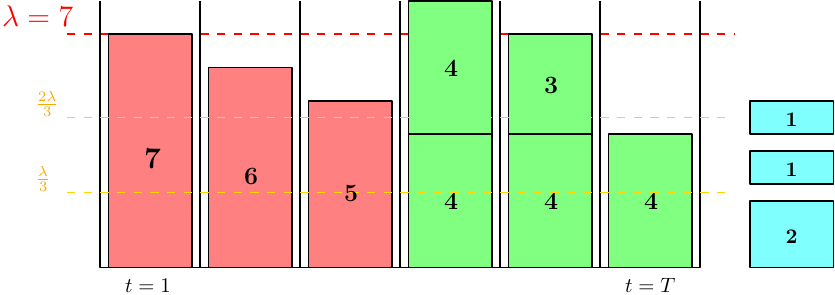}
    \caption{The result of \lpt\ after filling periods with big and medium items first. Items of sizes 1 and 2 ($p \le \frac{\lambda}{3}$) have still to be treated.}
    \label{fig:approx}
\end{figure}

Figure~\ref{fig:approx} illustrates an example of item assignment with \lpt\ after considering big (drawn in red) and medium (in green) items only. We fix $\lambda = 7$, $T = 6$, and items sizes $(7,6,5,4,4,4,4,3,2,1,1)$. The small items, in blue, are not put inside the boxes yet.

The end of the proof consists in applying again the argument used in the proof of Theorem~\ref{th:approx_1}. Now, all periods contain at most $\frac{4H^*}{3}$ resources: the periods with one big item do not exceed $H^*$ and the periods with at most two medium items do not exceed $\frac{4H^*}{3}$. The remaining items have size at most $\frac{H^*}{3}$, and $\vol(\bm{\pi}) \le \vol(\Piopt)$. So, at any step, the least filled period will contain at most $H^*$ resources in order to satisfy the volume condition: we can pack a small item into it without exceeding the bound $\frac{4H^*}{3}$.
\end{proof}

\section{Generation of packings with small volume} \label{sec:packings}

The objective of this section is to generate suitable packings for approximating \Mbot . More precisely, we aim at fixing the process of Step 1 of our algorithm \textsc{Bot-approx}. As stated in Theorem~\ref{th:collection}, we want to produce in polynomial-time a collection $\Pi$ of packing with certain requirements. This section is dedicated to the proof of Theorem~\ref{th:collection}. The conclusion of this proof is given in Section~\ref{subsec:proof}.

We focus on the specific formulation of \Mbot\ where $T=1$. The objective of this problem is to determine the minimum number of resources needed to process all jobs - at least $d_k$ jobs of type $k$ for all $k \in \K$ - in one single period, without reconfiguration. We denote it by \MbotOne. 
Its mathematical formulation can be obtained by simply replacing $T=1$ into the one given in Problem~\ref{pb:multibot}. A solution is thus a packing $(x_{pk})_{p,k}$ as the parameter $t$ does not intervene anymore here. The goal is to minimize the total number of resources, {\em i.e.} $\sum_{k\in \K} \sum_{p\in \p} p\cdot x_{pk}$, while all demands are satisfied, {\em i.e.} $\sum_{p\in \p} c_{pk}\cdot x_{pk} \ge d_k$ for any $k \in \K$. 

Concretely, solving \MbotOne\ provides us with the optimal way to process all jobs during only one session. From Section~\ref{subsec:multibot}, we know that the input size of \MbotOne\ is $O(KP(\log \cmax + \log \dmax))$. We will see in the remainder that an optimal packing for \MbotOne\ can be found in polynomial time. 

\subsection{Optimal packings with limited volume and scale}

To deal with the requirements of Theorem~\ref{th:collection}, we define a problem which is more general than \MbotOne . We call it \MbotOnel{$\lambda$}{$\tau$}: it makes two extra parameters $\lambda, \tau \in \mathbb{N}$ intervene, see Problem~\ref{pb:multibot1}.   

We describe the non-trivial constraints presented in Problem~\ref{pb:multibot1}.
Constraint \eqref{ILP2.1} means that we must have a sufficient capacity to satisfy all demands $d_k$. 
Constraint \eqref{ILP2.2} means that the volume of the packing must be at most $\lambda T$. Constraint \eqref{ILP2.3} means that a performed job must contain no more than $\lambda$  resources. 
Finally, constraint \eqref{ILP2.4} implies that the $\lambda$-scale of the resulting packing must be less than $\tau$. Without constraint~\eqref{ILP2.2}, the problem would always admit a solution, but its volume $H$ could be as large as possible. Here, we are only interested in solutions with a volume bounded by some polynomial function.

Observe that if $\lambda \rightarrow +\infty$, constraints~\eqref{ILP2.2}, \eqref{ILP2.3} and~\eqref{ILP2.4} disappear and the problem ``tends to'' \MbotOne, which always admit a solution. In the remainder, we abuse notation and denote by $\lambda = \infty$  this case. Our idea consists in solving \MbotOnel{$\lambda$}{$\tau$} for $\tau = T$ and all values $\lambda \in \set{1,2,3,\ldots,3P-1,\infty}$: if we find a solution, it will be put into collection $\Pi$. 

\begin{problem}[\MbotOnel{$\lambda$}{$\tau$}]
\begin{align}
 \mbox{\emph{\textbf{Input:}}} ~~ & \K = \set{1,\ldots,K}, \p = \set{1,\ldots,P} \nonumber \\
 & \emph{Capacities}~  (c_{pk})_{\substack{p \in \p \\ k \in \K}} \nonumber \\
 & \emph{Demands}~   (d_k)_{k \in \K} \nonumber\\
  &  \mathbf{\lambda}, \tau \in \mathbb{N}, 1\le \lambda\le 3P \nonumber \\
 \mbox{\emph{\textbf{Objective:}}} ~~   & \emph{find a packing}~ \bm{\pi}=(x_{pk})_{\substack{p \in \p \\ k \in \K}} ~\emph{which} & \nonumber
 \\
 & \emph{minimizes}~ H ~\emph{\textbf{or} answer NO} & \nonumber \\
        \mbox{\emph{\textbf{subject to:}}} ~~ & \sum_{p\in\p} c_{pk} \cdot x_{pk} \geq d_k & \forall k \in \K \label{ILP2.1}\\
     & \sum_{k\in\K} \sum_{p\in \p} p \cdot x_{pk} = H \le \lambda \tau &  \label{ILP2.2}\\
     & x_{pk} = 0   & \forall k \in \K, p > \lambda  \label{ILP2.3}\\ 
     & \sum_{k\in \K} \sum_{p > \frac{2\lambda}{3}} x_{pk} + \frac{1}{2}\sum_{k\in \K} \sum_{\frac{\lambda}{3} < p \le \frac{2\lambda}{3}} x_{pk} \leq \tau &  \label{ILP2.4}\\ 
     & x_{pk}, H \in \mathbb{N} \mbox{, }  & \forall k \in \K, p \in \p    \label{ILP2.5}
\end{align}
\label{pb:multibot1}
\end{problem}

Concretely, let $\bm{\pi}(\lambda)$ denote an optimum packing for \MbotOnel{$\lambda$}{$T$} if it exists. We define: 
\begin{equation}
    \Pi = \set{\bm{\pi}(\lambda) : 1\le \lambda\le 3P-1 ~\mbox{or}~ \lambda = \infty}
\label{eq:collection}
\end{equation}

Collection $\Pi$ will contain at least 1 packing because $\bm{\pi}(\infty)$ necessarily exists, and at most $3P$ packings. We prove now that the optimum solutions for all these problems produce the expected collection of packings. 

\begin{theorem}
Packings $\bm{\pi}(\lambda)$ satisfy the following properties:
\begin{enumerate}
    \item for any $\lambda$, $\bm{\pi}(\lambda)$, if it exists, reaches all demands $d_k$,
    \item $\volit(\bm{\pi}(\infty)) \le \volit(\Piopt)$,
    \item if $H^* < 3P$, then $\bm{\pi}(H^*)$ exists and we have: $\volit(\bm{\pi}(H^*)) \le \volit(\Piopt)$, $\max(\bm{\pi}(H^*)) \le H^*$, $\scait{H^*}(\bm{\pi}(H^*)) \le~T$.
\end{enumerate}
\label{th:opt_packings}
\end{theorem}
\begin{proof}
1. The constraint~\eqref{ILP2.1} implies that any solution achieves all demands, independently from value $\lambda$.

2. Observe that the solutions (not necessarily optimal) of $\MbotOnel{\infty}{T}$ correspond exactly to the set of packings satisfying all demands - constraint~\eqref{ILP2.1}. Indeed, constraints~\eqref{ILP2.3} and~\eqref{ILP2.4} disappear when $\lambda = \infty$. In particular, $\Piopt$ is one of these solutions and hence $\vol(\bm{\pi}(\infty)) \le \vol(\Piopt)$.

3. The reasoning is similar: if $H^* < 3P$, then $\Piopt$ is a solution of\\ $\MbotOnel{H^*}{T}$. Obviously it satisfies all demands and its volume is at most $TH^*$. Furthermore, as each period in schedule $\mathbf{x^*}$ contains at most $H^*$ resources, {\em i.e.} $H_t^* = \sum_{p,k} x_{pkt}^* \le H^*$, then necessarily $x_{pkt} = 0$ when $p > H^*$. So, $x_{pk} = 0$ for $p > H^*$ and any $k \in \mathcal{K}$. Finally, each period of schedule $\mathbf{x}^*$ cannot contain both a big configuration ($p > \frac{2H^*}{3}$) and a medium one ($\frac{H^*}{3} < p \le \frac{2H^*}{3}$). It cannot contain either three medium ones as it overpasses capacity $H^*$. As a conclusion, each period contains at most either a big configuration or two medium ones: the $H^*$-scale of $\Piopt$ is at most $T$. Hence, $\MbotOnel{H^*}{T}$ admits at least one solution $\Piopt$, so $\bm{\pi}(H^*)$ exists. As $\bm{\pi}(H^*)$ is the solution minimizing the volume, we have $\vol(\bm{\pi}(H^*)) \le \vol(\Piopt)$. The two other inequalities are the consequences of the definition of $\MbotOnel{H^*}{T}$.
\end{proof}

The collection $\Pi$ meets the requirements of Theorem~\ref{th:collection}. Now, we prove that all these packings can be built in polynomial time.

\subsection{Dynamic programming for the single-period problems} \label{subsec:dp}

We begin with the generation of packings $\bm{\pi}(\lambda)$ for $1\le \lambda\le 3P-1$. We will fix the case $\lambda = \infty$ in Section~\ref{subsec:opt_inf}.
We fix some positive integer $\lambda$ with $\lambda \le 3P-1$. Our objective is to solve $\MbotOnel{\lambda}{T}$ and, hence, to obtain either some packing $\bm{\pi}(\lambda)$ or a negative answer. We present a dynamic programming (DP) procedure to achieve this task. 

\textbf{Structure of DP memory}. From now on, variable $\tau$ represents a positive half-integer upper-bounded by $T$ and $W$ a positive integer representing the authorized volume, which is at most $\lambda T$.

We construct a four-dimension vector \texttt{Table}, whose elements are $\tablebot{p,k,W,\tau}$ with three integers $0\le p\le \lambda$, $1\le k\le K$, $0\le W\le \lambda T$ and a half-integer $0\le \tau \le T$. Hence, the total size of the vector is $\lambda^2T^2K = O(P^2T^2K)$.
The role of this vector is to help us producing intermediary packings which allows us to obtain potentially the solution of $\MbotOnel{\lambda}{T}$.
We denote by $\bm{\pi}[p,k,W,\tau]$ a packing which:
\begin{itemize}
\item produces only jobs of type $1,2,\ldots,k$: $x_{pk'}= 0 $ when $k' > k$,
\item uses only configurations $1,2,\ldots, p$ to process the jobs of type $k$: $x_{p'k} = 0$ when $p' > p$,
\item satisfies the demands until $d_{k-1}$: $\sum\limits_{p'=1}^{\lambda} c_{p'k'}\cdot x_{p'k'} \ge d_{k'}$ for all $1\le k'< k$,
\item has a volume at most $W$: $\sum\limits_{p' = 1}^{\lambda} \sum\limits_{k' = 1}^{k} p'\cdot x_{p'k'} \le W$,
\item has a $\lambda$-scale at most $\tau$: $\sca{\lambda}(\bm{\pi}) \le \tau$,
\item maximizes the number of jobs of type $k$ processed, {\em i.e.} $\sum\limits_{p'=1}^{\lambda} c_{p'k}\cdot x_{p'k}$.
\end{itemize}

Packing $\bm{\pi}[p,k,W,\tau]$ does not necessarily exist since reaching the demands might be avoided by the volume and scale conditions. Observe that, by definition, $\MbotOnel{\lambda}{T}$ admits a solution if and only if $\bm{\pi}[\lambda,K,\lambda T,T]$ exists and satisfies the demands for jobs of type $K$. Also, for the specific value $p = 0$, the packing $\bm{\pi}[0,k,W,\tau]$ cannot use any configuration for jobs of type $k$, so in fact it does not process jobs of type $k$ at all. We can fix: $\bm{\pi}[0,k,W,\tau] = \bm{\pi}[\lambda,k-1,W,\tau]$. The particular case $p = 0$ and $k = 1$ gives an empty packing $\bm{\pi}[0,1,W,\tau]$: it will be forgotten in the remainder.

The objective of vector \texttt{Table} is to contain either the production of jobs of type $k$ by $\bm{\pi}[p,k,W,\tau]$ if it exists, or $-\infty$. More formally,
\begin{equation}
\tablebot{p,k,W,\tau} = \left\{
\begin{split}
& \sum_{p'=1}^p c_{p'k} x_{p'k} & \mbox{if}~ \bm{\pi}[p,k,W,\tau] = (x_{p'k'})_{p',k'} ~\mbox{exists}\\
& -\infty  & \mbox{otherwise}\\
\end{split}
\right.
\label{eq:table_expected}
\end{equation}

In addition, we also compute two vectors \texttt{Bool} and \texttt{Pack} with the same dimension sizes than \texttt{Table}. Vector \texttt{Bool} simply indicates, with a boolean, whether $\bm{\pi}[p,k,W,\tau]$ satisfies the demand $d_k$ for jobs of type $k$ ($\boolbot{p,k,W,\tau} = \True$) or not ($= \False$). Finally, vector \texttt{Pack} provides us with some information which allows us to retrieve exactly the composition of $\bm{\pi}[p,k,W,\tau]$. More details on vector \texttt{Pack} will follow.

Before stating our recursive formula to achieve the DP algorithm, we define an extra function $S_{p,\lambda}$ for any $1\le p\le \lambda$. Given some half-integer ``budget'' $\tau$ and an integer $j$, it returns the updated $\lambda$-scale budget which remains after adding a number $j$ of $p$-resources to some packing. Formally,
\begin{equation}
S_{p,\lambda}(\tau,j) = \left\{
\begin{split}
& \tau & \mbox{if}~ p \le \frac{\lambda}{3}\\
& \tau - \frac{j}{2}  & \mbox{if}~ \frac{\lambda}{3} < p \le \frac{2\lambda}{3}\\
& \tau - j & \mbox{if}~ p > \frac{2\lambda}{3}\\
\end{split}
\right.
\nonumber
\end{equation}

\textbf{Recursive formula}. We state now a recursive formula for computing all values of \texttt{Table}, \texttt{Bool} and \texttt{Pack}. We begin with the recursive scheme of \texttt{Table} and \texttt{Bool}. 

The base case is $p = k = 1$. Packing $\bm{\pi}[1,1,W,\tau]$ has no demand to satisfy, therefore it necessarily exists. It corresponds to taking the maximum number of $1$-resources which satisfy both the volume and scale conditions:
\begin{equation}
\forall W, \forall \tau, ~~~\tablebot{1,1,W,\tau} = \max_{\substack{0\le j\le W\\ S_{1,\lambda}(\tau,j) \ge 0}} j\cdot c_{11}
\label{eq:tab_base}
\end{equation}
\begin{equation}
\boolbot{1,1,W,\tau} = \True \Leftrightarrow \tablebot{1,1,W,\tau} \ge d_1
\label{eq:bool_base}
\end{equation}

For the general statement, we distinguish two cases. First assume that $p = 0$ and $k > 1$. We have $\bm{\pi}[0,k,W,\tau] = \bm{\pi}[\lambda,k-1,W,\tau]$.
\begin{equation}
\tablebot{0,k,W,\tau}
= \left\{
\begin{split}
& 0 & \mbox{if}~ \tablebot{\lambda,k-1,W,\tau} \ge d_{k-1}\\
& -\infty & \mbox{otherwise}\\
\end{split}
\right.
\label{eq:table_zero}
\end{equation}
\begin{equation}
\boolbot{0,k,W,\tau} = \False
\label{eq:bool_zero}
\end{equation}

Second, assume that $p \ge 1$, {\em i.e.} jobs of type $k$ can be processed by $p'$-resources, with $p' \le p$. We write:

\begin{equation}
\tablebot{p,k,W,\tau} = \hspace{-1.2cm}\max_{\substack{0\le j\le \lfloor \frac{W}{p} \rfloor\\ S_{p,\lambda}(\tau,j) \ge 0\\ \boolbot{\lambda,k-1,W-jp,S_{p,\lambda}(\tau,j)}}}  \tablebot{p-1,k,W-jp,S_{p,\lambda}(\tau,j)} + j c_{pk}
\label{eq:table_rec}
\end{equation}
\begin{equation}
\boolbot{p,k,W,\tau} = \True \Leftrightarrow \tablebot{p,k,W,\tau} \ge d_k
\label{eq:bool_rec}
\end{equation}

Index $j$ represents the number of $p$-resources we might add to our packing. Naturally, this addition should overpass neither the volume constraint ($j \le \frac{W}{p}$), nor the scale one ($S_{p,\lambda}(\tau,j) \ge 0$). Furthermore, $\boolbot{\lambda,k-1,W-jp,S_{p,\lambda}(\tau,j)}$ must be True as it guarantees that demands $d_1,\ldots,d_{k-1}$ can be satisfied before the add-on of $j$ $p$-resources. It may happen that no index $j$ (even $j = 0$) satisfies these three conditions. In this case, we fix $\tablebot{p,k,W,\tau} = -\infty$.


Observe that $\tablebot{p,k,W,\tau} = -\infty$ if and only if $\boolbot{\lambda,k-1,W,\tau}$ is \False. This makes sense since it means that with the volume $W$ and the scale $\tau$ we are considering, the demands for jobs of type $1,2,\ldots,k-1$ could not be reached. Figure~\ref{fig:dp_calls} provides us with a 2D-view of vector \texttt{Table}, highlighting the area where values are base cases and giving examples of recursive calls, following Equations~\eqref{eq:table_zero} and~\eqref{eq:table_rec}. We prove that $\tablebot{p,k,W,\tau}$ is equal to the expectations expressed in Equation~\eqref{eq:table_expected}.

\begin{figure}[t]
    \centering
    \includegraphics[scale=0.8]{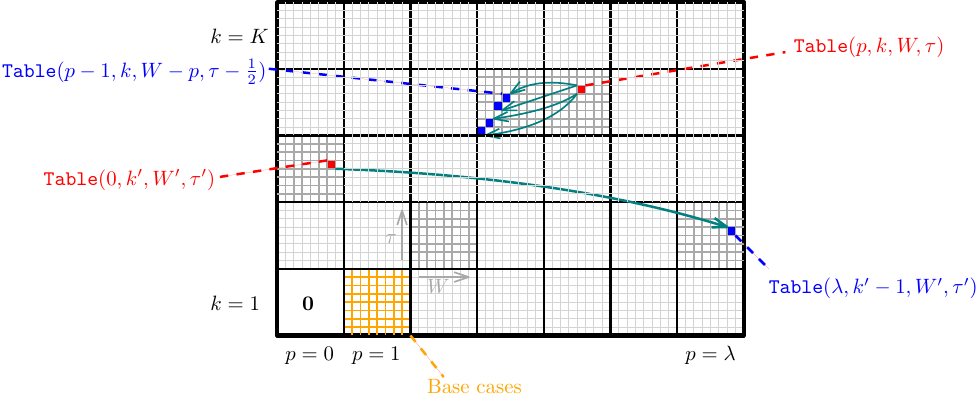}
    \caption{2D-projection of the 4D-vector \texttt{Table} illustrating the different recursive calls}
    \label{fig:dp_calls}
\end{figure}

\begin{lemma}
If $\bm{\pi}[p,k,W,\tau] = (x_{p'k'})_{p',k'}$ exists, then:
\begin{itemize}
\item $\emph{\texttt{Table}}[p,k,W,\tau] = \sum_{p'} c_{p'k} x_{p'k}$. 
\item $\emph{\texttt{Bool}}[p,k,W,\tau] = \emph{True}$ if and only if $\bm{\pi}[p,k,W,\tau]$ satisfies demand $d_k$. 
\end{itemize}
Otherwise, $\emph{\texttt{Table}}[p,k,W,\tau] = -\infty$ and $\emph{\texttt{Bool}}[p,k,W,\tau] = \emph{False}$. 
\label{le:table}
\end{lemma}
\begin{proof}
First, we handle the base case. If $p = k = 1$, then $\bm{\pi}[1,1,W,\tau]$ exists and is the singleton packing $(x_{11})$ which processes the maximum number of jobs of type $1$ with $1$-resources while satisfying both the volume and scale constraints. Then, $\tablebot{1,1,W,\tau}$ must be equal to $x_{11} c_{11}$, which perfectly fits with the definition given in Equation~\eqref{eq:tab_base}. Moreover, according to Equation~\eqref{eq:bool_base}, $\boolbot{1,1,W,\tau}$ is \True\ if and only if $\bm{\pi}[1,1,W,\tau]$ reaches demand $d_k = d_1$.

We proceed inductively. Assume that the statement is achieved not only for all pairs $(p',k')$ such that $k' < k$ but also when $k' = k'$ and $p' < p$.

If $\bm{\pi}[p,k,W,\tau]$ does not exist, then it means that demands $d_1,\ldots,d_{k-1}$ cannot be satisfied with volume $W$ and scale $\tau$ constraints. So, $\tablebot{\lambda,k-1,W,\tau} < d_{k-1}$ and $\boolbot{\lambda,k-1,W,\tau}$ is \False\ by induction. Consequently, Equations~\eqref{eq:table_zero} and~\eqref{eq:table_rec} ensure us that $\tablebot{p,k,W,\tau} = -\infty$ and $\boolbot{p,k,W,\tau} = \False$, as expected. In the remainder of this proof, we assume $\bm{\pi}[p,k,W,\tau]$ exists.

If $p = 0$, then $\bm{\pi}[0,k,W,\tau] = \bm{\pi}[\lambda,k-1,W,\tau]$. As this packing does not process any job of type $k$, we fixed $\tablebot{p,k,W,\tau} = 0$, see Equation~\eqref{eq:table_zero}. Hence, as stated in Equation~\eqref{eq:bool_zero}, $\bm{\pi}[0,k,W,\tau]$ obviously does not satisfy demand $d_k>0$.

If $p > 0$, then $\bm{\pi}[p,k,W,\tau]$ is made up of a certain number $j \ge 0$ of $p$-resources which process jobs of type $k$, and all other components which can be seen together as a slightly smaller packing.  The latter packing must satisfy the demands for jobs of type $1,\ldots,k-1$ and also maximize the number of processed jobs of type $k$ with configurations $1,\ldots,p-1$, while fulfilling the volume and scale constraints, even if we know that already a number $j$ of $p$-resources have to be counted. Hence, $\bm{\pi}[p,k,W,\tau]$ can be obtained from $\bm{\pi}[p-1,k,W-jp,S_{p,\lambda}(\tau,j)]$ by adding only $x_{pk} = j$. This justifies Equation~\eqref{eq:table_rec}. Finally, as stated in Equation~\eqref{eq:bool_rec}, $\boolbot{p,k,W,\tau}$ is \True\ if and only if demand $d_k$ is satisfied by $\bm{\pi}[p,k,W,\tau]$, {\em i.e.} $\tablebot{p,k,W,\tau} \ge d_k$.
\end{proof}

Values $\tablebot{p,k,W,\tau}$ provide us with the number of jobs of type $k$ processed by optimum packings $\bm{\pi}[p,k,W,\tau]$. Nevertheless, our initial objective is to obtain the composition of these packings. To achieve it, we create a third table \texttt{Pack} with the same dimensions than \texttt{Table} and \texttt{Bool}. A first natural idea would be to fill each element $\packbot{p,k,W,\tau}$ with the packing $\bm{\pi}[p,k,W,\tau]$. However, as the encoding size of $\bm{\pi}[p,k,W,\tau]$ is at least $PK$, it will have some relatively strong impact on the complexity of our algorithm. Hence, we fill $\packbot{p,k,W,\tau}$ only with the necessary information needed to retrieve the packings. 

If $\tablebot{p,k,W,\tau} = -\infty$ or $p = 0$, then $\packbot{p,k,W,\tau}$ is kept empty. Else, we define $\packbot{p,k,W,\tau}$ as the integer $j$ which is the number of $p$-resources processing jobs of type $k$ in $\bm{\pi}[p,k,W,\tau]$.

\begin{equation}
\packbot{p,k,W,\tau} = \left\{
\begin{split}
& j ~\mbox{if}~ p=k=1 ~\mbox{and}~ \tablebot{1,1,W,\tau} = j c_{11}\\
& j ~\mbox{if}~ \tablebot{p,k,W,\tau} = \tablebot{p-1,k,W-jp,S_{p,\lambda}(\tau,j)} + j c_{pk}\\
& \mbox{empty otherwise}
\end{split}
\right.
\label{eq:pack}
\end{equation}

Observe that the encoding size of each value $\packbot{p,k,W,\tau}$ is now $O(\log (\lambda T))$. Some packing $\bm{\pi}[p,k,W,\tau]$ can now be recovered with the following process, we call \textsc{Recover}:
\begin{itemize}
\item if $\tablebot{p,k,W,\tau} = -\infty$, then $\bm{\pi}[p,k,W,\tau]$ does not exist.
\item if $p = k = 1$ and $\packbot{1,1,W,\tau} = j$, then $\bm{\pi}[1,1,W,\tau]$ is the singleton $(x_{11})$ with $x_{11} = j$.
\item if $p = 0$ and $k > 1$, then $\bm{\pi}[0,k,W,\tau] = \bm{\pi}[\lambda,k-1,W,\tau]$.
\item if $\packbot{p,k,W,\tau} = j$, compute recursively packing $\bm{\pi}[p-1,k,W-jp,S_{p,\lambda}(\tau,j)]$ with \textsc{Recover} and add $x_{pk} = j$ to it. 
\end{itemize}

\begin{algorithm}[h]
\SetKwFor{For}{for}{do}{\nl endfor}
\SetKwFor{Forall}{for all}{do}{\nl endfor}
\SetKwIF{If}{ElseIf}{Else}{if}{then}{else if}{else}{\nl endif}
\DontPrintSemicolon
\SetNlSty{}{}{:}
\SetAlgoNlRelativeSize{0}
\SetNlSkip{1em}
\nl\KwIn{Instance of $\MbotOnel{\lambda}{T}$}
\nl\KwOut{A packing $\bm{\pi}(\lambda)$ or answer NO}
\nl Initialize $\texttt{Table}$, $\texttt{Bool}$ and $\texttt{Pack}$ as empty vectors;\;
\nl \For{all $1\le W\le \lambda T$ and $1\le \tau \le T$}{
	\nl fix $\tablebot{1,1,W,\tau}$, $\boolbot{1,1,W,\tau}$, $\packbot{1,1,W,\tau}$ with respectively Equations~\eqref{eq:tab_base},~\eqref{eq:bool_base}, and~\eqref{eq:pack};\;
}
\nl $p = 1$;\;
\nl \For{all $1\le k\le K$}{
	\nl \While{$p \le \lambda$}{
		\nl \For{all $1\le W\le \lambda T$ and $1\le \tau \le T$}{
			\nl fix $\tablebot{p,k,W,\tau}$, $\boolbot{p,k,W,\tau}$, $\packbot{p,k,W,\tau}$ with Equations~(\ref{eq:table_zero}-\ref{eq:pack});\;
		}
		\nl $p = p + 1$
	}	
	\nl $p = 0$;\;
}
\nl $D = \tablebot{\lambda,K,\lambda T, T}$;\;
\nl \lIf{$D \ge d_K$}{
	\Return \textsc{Recover}$(\lambda,K,\lambda T, T)$
}
\nl \Return NO
\caption{Exact algorithm \textsc{Bot-ProgDyn} which solves $\MbotOnel{\lambda}{T}$}
\label{algo:progdyn}
\end{algorithm}

We remind that if $\tablebot{\lambda,K,\lambda T,T} \ge 0$, then $\bm{\pi}[\lambda,K,\lambda T,T]$ exists and it gives us some $\bm{\pi}(\lambda)$ if and only if the demand $d_K$ is achieved. Otherwise, if $\tablebot{\lambda,K,\lambda T,T} = -\infty$, then there is no packing $\bm{\pi}(\lambda)$ for sure. As a consequence, the following algorithm solves $\MbotOnel{\lambda}{T}$: first compute simultaneously the three tables \texttt{Table}, \texttt{Bool} and \texttt{Pack}, second answer NO if $\tablebot{\lambda,K,\lambda T,T} = -\infty$, otherwise retrieve packing $\bm{\pi}[\lambda,K,\lambda T,T]$ recursively with sub-routine \textsc{Recover}. If $\bm{\pi}[\lambda,K,\lambda T,T]$ satisfies demand $d_K$, then return it as $\bm{\pi}(\lambda)$, else answer NO. We call this algorithm \textsc{Bot-ProgDyn} (Algorithm~\ref{algo:progdyn}).

\begin{theorem}
\textsc{Bot-ProgDyn} solves $\MbotOnel{\lambda}{T}$ with a running time $O(\lambda^3T^3K)$.
\label{th:dp_lambda}
\end{theorem}
\begin{proof}
From Lemma~\ref{le:table}, we know that $\tablebot{\lambda,K,\lambda T, T} \ge d_K$ if and only if $\bm{\pi}(\lambda)$ exists. Therefore, as \textsc{Bot-ProgDyn} achieves the dynamic programming to fills \texttt{Table}, it returns the solution of $\MbotOnel{\lambda}{T}$. 

We focus now on the running time of \textsc{Bot-ProgDyn}. The dyanmic programming routine uses a memory space of size $O(\lambda^2T^2K)$, and each computation uses at most $\lambda T$ recursive calls: the worst case occurs with Equation~\eqref{eq:table_rec} and the computation of $\tablebot{\lambda,K,\lambda T, T}$. Consequently, the number of comparisons/affectations/arithmetic operations is $O(\lambda^3T^3K)$.

Furthermore, the time needed to apply \textsc{recover} is negligible compared to the latter. Indeed, the number of recursive calls needed to compute $\bm{\pi}(\lambda,K,\lambda T, T)$ is at most the size of vector \texttt{Pack}, which is $O(\lambda^2T^2K)$.
\end{proof}

\subsection{Optimum packing with unlimited scale} \label{subsec:opt_inf}

We focus on problem $\MbotOne$, which is equivalent to $\MbotOnel{\infty}{T}$ ($T$ has no influence here). In other words, we aim at producing the packing $\bm{\pi}(\infty)$ which will be denoted $\bm{\pi}$ in this subsection to simplify notations.

Our reasoning is based on the result stated in Lemma~\ref{le:optim_config}. We will assume that the packing $\bm{\pi} = (x_{pk})_{p,k}$ admits low values for non-optimum configurations, {\em i.e.} $x_{pk} \le p_0(k)$ for any $k \in \mathcal{K}$ and $p \neq p_0(k)$ (consequence of Lemma~\ref{le:optim_config}). This packing $\bm{\pi}$ can be decomposed into two parts: on one hand, the values $x_{p_0(k),k}$ for optimal configurations and, on the other hand, a packing $\bm{\pi}' = (x_{pk}')_{p,k}$ taking only values for non-optimum configurations. Formally,

\begin{equation}
x_{pk}' = \left\{
\begin{split}
& x_{pk} & \mbox{if}~ k \in \K, p \in \p, p\neq p_0(k)\\
& 0  & \mbox{if}~ k \in \K, p = p_0(k)\\
\end{split}
\right.
\nonumber
\end{equation}

This sub-packing, as expected, has a small volume. Let $W(P) = P\frac{P(P-1)}{2}$. We have $\vol(\bm{\pi}') = \sum_{p,k} p\cdot x_{pk}' \le \sum_{k} \sum_{p \neq p_0(k)} p\cdot p_0(k) \le K\cdot W(P).$
More precisely, the volume of $\bm{\pi}'$ dedicated to each type of job is at most $W(P)$.

Our construction of packing $\bm{\pi}$ is inductive. For any $1\le k\le K$, we denote by $\bm{\pi}_j = (x_{pk}^{(j)})_{p,k}$ a packing satisfying the following properties:
\begin{itemize}
\item $\bm{\pi}_j$ satisfies the demands for all types of jobs,
\item for any $1\le i\le j$, $\bm{\pi}_j$ uses the minimum number of resources to satisfy demand $d_i$ with all configurations allowed,
\item for any $j < i \le K$, uses the minimum number of resources to satisfy demand $d_i$ with only configuration $p_0(j)$ allowed.
\end{itemize}

First, observe that $\bm{\pi}_K$ is a solution of $\MbotOne$ since it satisfies all demands while minimizing the total volume of the packing. Therefore, we can state $\bm{\pi} = \bm{\pi}_K$. Furthermore, $\bm{\pi}_0$ is a packing satisfying all demands which minimizes the total volume by using only optimal configurations for each type of jobs. It can be determined analytically and this will be the base case of our induction. For any $1\le k\le K$,
$$x_{p_0(k),k} = \left\lceil \frac{d_k}{c_{p_0(k),k}} \right\rceil$$

By definition, all other components of packing $\bm{\pi}_0$ are zero. Below, we proceed the inductive step.
\begin{lemma}
Assume packing $\bm{\pi}_j$ is known. One can build packing $\bm{\pi}_{j+1}$ in time $O(W(P)^2P)$ using the algorithm \textsc{Bot-ProgDyn}.
\label{le:induction}
\end{lemma}
\begin{proof}
We remind that we are working on some instance $\mathcal{I}$ of $\MbotOne = \MbotOnel{\infty}{0}$ (we remind that when $\lambda = \infty$, then the second parameter has no influence on the definition of the problem, so we put 0 arbitrarily). We define another instance $\mathcal{J}$ of $\MbotOnel{\infty}{0}$. Instance $\mathcal{J}$ contains only one job type, which corresponds to the jobs of type $j+1$ in $\mathcal{I}$. Then, the capacities in $\mathcal{J}$ are exactly the capacities $c_{p(j+1)}$ of jobs of type $j +1$ in $\mathcal{I}$, except for the optimal configuration for which we fix $c_{p_0(j+1),j+1} = 0$. Concretely, $p_0(j+1)$ will not be used in the solutions of $\mathcal{J}$. We do not need to fix any demand constraint because, with only one job type, it will not intervene in algorithm \textsc{Bot-ProgDyn}.

Using \textsc{Bot-ProgDyn} on instance $\mathcal{J}$, we build vectors \texttt{Table} and \texttt{Pack} (with exactly one job type, \texttt{Bool} is not necessary) with limited volume $W(P)$. In brief, we stop the procedure when $\tablebot{P,1,W(P),0}$ and $\packbot{P,1,W(P),0}$ are obtained. Thanks to these vectors, we are able to obtain, for any volume $0\le W\le W(P)$, a packing which processes only jobs of type $j+1$, with volume at most $W$, and which maximizes the production.

By induction hypothesis, we know that packing $\bm{\pi}_j$ already minimizes the volume used to process the necessary demand for jobs of type $1,\ldots,j$. Thus, we focus only on jobs of type $j+1$. Moreover, the volume $W$ used to process jobs of type $j+1$ in packing  $\bm{\pi}_{j+1}$ with non-optimal configurations does not exceed $W(P)$ (consequence of Lemma~\ref{le:optim_config}). Hence, we compute for all possible volumes $0 \le W \le W(P)$, the packing which maximizes the production of jobs of type $j+1$ while maintaining at most volume $W$. The production of such packing is given by $\tablebot{P,1,W,0}$. It suffices to guess the volume $W$ taken by non-optimal configurations, to compute the packing maximizing the production  of jobs of type $j+1$ with this volume and, eventually, to add optimal configurations $p_0(j+1)$ while the demand is not satisfied. At least one of these volumes $W$ will provide us with a packing $\bm{\pi}_{j+1}$ satisfying the properties stated above. Formally, we define:
\begin{equation}
W_{j+1} = \min_{0\le W\le W(P)} W + \left\lceil \frac{d_{j+1} - \tablebot{P,1,W,0}}{c_{p_0(j+1),(j+1)}} \right\rceil.
\label{eq:min_volume}
\end{equation}
The right-hand side of Equation~\eqref{eq:min_volume} is the total volume used for processing jobs of type $j+1$ when the demand has to be satisfied and also when exactly $W$ resources are assigned to non-optimum configurations.

Once the optimum volume $W_{j+1}$ for non-optimum configurations was determined, we fix, for the optimal configuration $p_0(j+1)$:
$$x_{p_0(j+1),j+1}^{(j+1)} = \left\lceil \frac{d_{j+1} - \tablebot{P,1,W_{j+1},0}}{c_{p_0(j+1),(j+1)}} \right\rceil$$
For $p \neq p_0(j+1)$, the number of $p$-resources assigned for the process of jobs of type $j+1$ is directly given by the packing which can be recovered from $\packbot{P,1,W_{j+1},0}$.
\end{proof}

\begin{theorem}
\MbotOne\ can be solved in time $O(KP^7)$. 
\label{th:packing_inf}
\end{theorem}
\begin{proof}
We state the full inductive process here. First, we compute packing $\bm{\pi}_0$ with analytical formulas in time $O(KP)$. Then, we construct inductively all packings $\bm{\pi}_j$ for $1\le j\le K$ thanks to Lemma~\ref{le:induction}. Eventually, packing $\bm{\pi}_K$ provides us with a solution of \MbotOne . As there are exactly $K$ inductive steps and each step runs in $O(P^7)$ according to Theorem~\ref{th:dp_lambda}, the total running time for computing $\bm{\pi}(\infty)$ is $O(KP^7)$.
\end{proof}

\subsection{Computation of collection $\Pi$} \label{subsec:proof}

Combining the results obtained in both Sections~\ref{subsec:dp} and~\ref{subsec:opt_inf}, we are now ready for the computation of the whole collection $\Pi$ of packings which will allow us to obtain a $\frac{4}{3}$-approximation for \Mbot .

\medskip

\textit{Proof of Theorem~\ref{th:collection}.}
According to Theorem~\ref{th:dp_lambda}, each packing $\bm{\pi}(\lambda)$, $1\le\lambda \le 3P-1$, can be produced in time $O(\lambda^3T^3K)$ (or the algorithm warns us that such packing does not exist). Then, according to Theorem~\ref{th:packing_inf}, packing $\bm{\pi}(\infty)$ can be computed in time $O(KP^7)$. Consequently, we build the collection described in Equation~\eqref{eq:collection} with the time announced.

By definition, any packing $\bm{\pi}(\lambda)$, if it exists, satisfies all demands (Problem~\ref{pb:multibot1}). Packing $\bm{\pi}(\infty)$, which necessarily exists, is the packing which satisfies all demands while minimizing the volume. As $\Piopt$ is a (not necessarily optimal) solution of \MbotOne, then its volume is at least $\vol(\bm{\pi}(\infty))$. Finally, if $H^* < 3P$, then packing $\bm{\pi}(H^*)$ exists since $\Piopt$ is a solution of $\MbotOnel{H^*}{T}$: it satisfies all demands, its volume is at most $TH^*$, its maximum at most $H^*$ and its scale is at most $T$.
\qed

\section{Conclusion and perspectives}

In this article, we present a $\frac{4}{3}$-approximation for a combinatorial problem - generalizing high-multiplicity \IMS\ - which consists in finding the optimum schedule for a fleet of reconfigurable robots given some demands for all types of processed jobs. We believe that this problem, we called \Mbot , could describe many situations involving reconfigurable systems, such as work plannings for teams of employees for example. Therefore, our approximation algorithm is, in our opinion, not only an important achievement for the industrial use of reconfigurable robots but also for other scheduling areas. Indeed, it offers the theoretical guarantee to be close to the optimum schedule, and in practice the performance should be in fact much smaller than this upper bound. 

A drawback of our algorithm could be the exponent 7 over parameter $P$ in its running time. For the reconfigurable robots application, it is not since in practice $P$ is smaller than $15$ (no more than 14 elementary robots of MecaBotiX can be assembled together). Nevertheless, this time complexity has to be taken into account and perhaps some improvements could be proposed.

A natural question is whether this approximation factor $\frac{4}{3}$ can be improved. The answer is yes, as we are currently writing a PTAS for \Mbot . This is a very interesting theoretical result, as it shows us that any approximation factor can be achieved, but, at the same time, it does not offer a performing running time for industrial application. This contribution will be presented soon.

A possibility would be to design the same algorithm by replacing \lpt\ with another heuristic for \IMS , for example \textsc{Multifit}. Indeed, the approximation ratio of \textsc{Multifit} is $\frac{13}{11}$. However, the analysis of \textsc{Multifit} is much more involved~\cite{yue90} than the one for \lpt . As a consequence, it is more difficult to see which are the properties our packings should satisfy in order to obtain a lower approximation ratio for \Mbot . Nevertheless, it will be the next direction of research to look at on this subject.

\bibliographystyle{plain}
\bibliography{biblio}

\end{document}